\documentclass{article}
\pdfoutput=1

\usepackage{hyperref}
\usepackage{authblk}

\usepackage{amsfonts}
\usepackage{multirow}
\usepackage{xcolor}

\usepackage{amsmath,amsthm}
\usepackage{amssymb}
\usepackage{tikz}
\usetikzlibrary{math,positioning,shapes,decorations.pathreplacing, calligraphy, backgrounds}
\usepackage{tikz-network}
\usepackage{multirow}

\makeatletter
\newcommand*{\centerfloat}{%
  \parindent \z@
  \leftskip \z@ \@plus 1fil \@minus \textwidth
  \rightskip\leftskip
  \parfillskip \z@skip}
\makeatother

\newtheorem{thm}{Theorem}

\newcommand{\els}{\mathbb{E}}
\newcommand{\reach}[2]{\mbox{$\mathcal{R}_{#1}(#2)$}}

\newcommand{\mret}{\textsc{mret}}

\newcommand{\tsat}{\textsc{3-sat}}
\newcommand{\nptime}{\mbox{NP}}

\newcommand{\true}{\mbox{\textsc{True}}}
\newcommand{\false}{\mbox{\textsc{False}}}

\begin{document}

\title{A Note on the Complexity of Maximizing Temporal Reachability via Edge Temporalisation of Directed Graphs}

\author[1]{Alkida Balliu}

\author[2]{Filippo Brunelli\thanks{This work was partly supported by the French National Research Agency (ANR) through through the projects  Multimod (ANR-17-CE22-0016) and Tempogral (ANR-22-CE48-0001).}}

\author[1]{Pierluigi Crescenzi}

\author[1]{Dennis Olivetti}

\newcommand\CoAuthorMark{\footnotemark[1]} 
\author[2]{Laurent Viennot\protect\CoAuthorMark}

\affil[1]{GSSI, I-67100 L'Aquila, Italy}
\affil[2]{Inria, Irif, Université de Paris, F-75013 Paris, France}

\maketitle

\begin{abstract}
A temporal graph is a graph in which edges are assigned a time label. Two nodes $u$ and $v$ of a temporal graph are connected one to the other if there exists a path from $u$ to $v$ with increasing edge time labels. We consider the problem of assigning time labels to the edges of a digraph in order to maximize the total reachability of the resulting temporal graph (that is, the number of pairs of nodes which are connected one to the other). In particular, we prove that this problem is NP-hard. We then conjecture that the problem is approximable within a constant approximation ratio. This conjecture is a consequence of the following graph theoretic conjecture: any strongly connected directed graph with $n$ nodes admits an out-arborescence and an in-arborescence that are edge-disjoint, have the same root, and each spans $\Omega(n)$ nodes.
\end{abstract}


\section{Introduction}

Temporal graphs have received increasing attention over the last two decades~\cite{Holme2012,Holme2013,Masuda2016,Michail2016} and have been defined in several different ways~\cite{Berman1996,Harary1997,Kempe2002Connectivity,BhadraF03,Cheng2003,Casteigts2012,Latapy2018} (see~\cite{Brunelli2021} for a classification of temporal graphs). Here, we say that a \textit{temporal graph} $G=(V,\els)$ is a list $\els$ of \textit{temporal edges} $(u,v,t)$, where $u,v\in V$ are two \textit{nodes} of the graph (called, respectively, \textit{tail} and \textit{head} of the temporal edge) and $t$ is the \textit{appearing time} of the temporal edge. For each temporal edge $(u,v,t)$, we can traverse the edge starting from $u$ at time $t$ and arrive in $v$ at time $t+1$, which is the \textit{arrival time} of the temporal edge.

We study a network optimisation problem related to the notion of reachability in temporal graphs. Given a temporal graph $G$, a (temporal) \textit{path} from a node $u$ to a node $v$ is a sequence $e_{1},e_{2},\ldots,e_{k}$ of temporal edges such that the tail of $e_{1}$ is $u$, the head of $e_{k}$ is $v$, and, for any $i$ with $1<i\leq k$, the tail of $e_{i}$ is equal to the head of $e_{i-1}$ and the appearing time of $e_{i}$ is greater than the appearing time of $e_{i-1}$. The \textit{temporal reachability} of $G$ is the number of pairs of nodes $u$ and $v$ such that $v$ is \textit{temporally reachable} from $u$, that is, there exists a temporal path from $u$ to $v$. The \textsc{Maximum Reachability Edge Temporalisation} (\mret) problem consists of, given a directed graph (in short, digraph) $D=(V,E)$, find an \textit{edge temporalisation} $\tau: E\rightarrow\mathbf{N}$ such that the temporal reachability of the resulting temporal graph is maximized. For example, let us consider the digraph shown in the left part of Figure~\ref{fig:temporalisation}. In the right part of the figure, we show an edge temporalisation of a digraph $D$ with four nodes, such that the temporal reachability of the resulting temporal graph is equal to $16$, which is clearly the maximum possible temporal reachability.

\begin{figure}[t]
\centering{\includegraphics[scale=1.0]{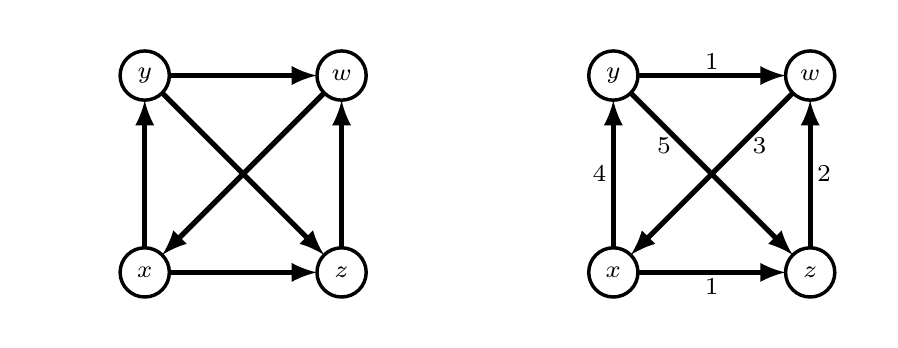}}
\caption{A digraph $D$ (left), and an edge temporalisation of $D$ (right). The temporal reachability of the resulting temporal graph is $16$, since each node is temporally reachable from any other node.}
\label{fig:temporalisation}
\end{figure}

The \mret\ problem restricted to undirected graphs has been studied in~\cite{Goebel1991}, where it is shown that the problem of deciding whether the resulting temporal graph is temporally connected (that is, for any two nodes $u$ and $v$, $v$ is temporally reachable from $u$) is \nptime-complete (clearly, this implies that the \mret\ problem restricted to undirected graphs is \nptime-hard). It is also easy to see that the \mret\ problem restricted to undirected connected graphs can be approximated within a constant approximation ratio, since this simply requires to look for a ``centroid'' in a spanning tree as a temporalisation where half of the nodes can reach the other half can then easily be computed.
Note, however, that temporalising a symmetric digraph is not equivalent to temporalising an undirected graph as different times can be assigned to an edge $(u,v)$ and the symmetric edge $(v,u)$.

In the following, we prove that the \mret\ problem is \nptime-hard, even when restricted to strongly connected digraphs. We will conclude by conjecturing that the problem is approximable within a constant approximation ratio, and suggesting a graph theory conjecture which could be use to proved the conjecture concerning the approximability of the \mret\ problem.

Problems similar to the one considered in this paper have already been analysed~\cite{Kempe2002Connectivity,MertziosMCS2013,Coro2019,Enright2019,Mertzios2019,Molter2021,Mertzios2021,Brunelli2022}. For instance, in~\cite{Mertzios2019} the authors propose two cost minimization parameters for temporal network design (that is, the maximum number of appearing times of an edge and the total number of appearing times of all edges), and they study the problem of optimizing these parameters subject to some connectivity constraint.

\section{Hardness result}

The next result shows that there is no polynomial-time algorithm solving the \mret{} problem, unless P is equal to NP. In the following, we will refer to edge temporalisations as schedules, that is, as orderings of the edges of the digraph. Indeed, one can easily transform an edge temporalisation $\tau$ into an edge temporalisation $\tau'$, where all time labels are pairwise distinct and where the total reachability according to $\tau$ is preserved. We then note that the total reachability according to $\tau'$ depends only on the ordering of the edges according to their time label. Given a digraph $D=(V,E)$ and a schedule $S$, a node $v$ is said to be $S$-reachable from a node $u$ if it is temporally reachable in the temporal graph $G$ induced by $D$ and the temporalisation $\tau_S$ that assigns appearing time $i$ to the $i$th edge of $S$ for $i\in [|E|]$. The set of nodes $S$-reachable from a node $u$ is denoted as $\reach{G}{u}$. The \textit{$S$-reachability} of $D$ is defined as the temporal reachability of the temporal graph induced by $\tau_S$.

\begin{thm}\label{thm:hardnessedgetemporalisation}
The \mret{} problem is NP-hard, even if the digraph $D$ is strongly connected.
\end{thm}

\begin{proof}
We reduce \tsat{} to \mret{} as follows. Let us consider a \tsat{} formula $\Phi$, with $n$ variables $x_1, \dots, x_n$ and $m$ clauses $c_1, \dots, c_m$. Without loss of generality we will assume that each variable appears positive in at least one clause and negative in at least one clause. We first define the \textit{unweighted} digraph $D=(V,E)$ as the union of the following gadgets (see Figure~\ref{fig:mret-reduction}).

\begin{description}
\item[\textbf{Variable gadgets}] For each variable $x_i$ of $\Phi$, $V$ contains the nodes $t_i^1$, $t_i^2$, $f_i^1$,and $f_i^2$ and $E$ contains the edges $(t_i^1,f_i^2)$, $(f_i^2,f_i^1)$, $(f_i^1,t_i^2)$, and $(t_i^2,t_i^1)$.

\item[\textbf{Clause gadgets}] For each clause $c_j$, $V$ contains the nodes $c_j^1$ and $c_j^2$. If the literal $x_i$ appears in $c_j$, $E$ contains the edges $(c_j^1,t_i^1)$ and $(t_i^2,c_j^2)$, while if the literal $\neg x_i$ appears in $c_j$, $E$ contains the edges $(c_j^1,f_i^1)$ and $(f_i^2,c_j^2)$. Moreover, for each two clauses $c_h$ and $c_j$ with $h \neq j$, $E$ contains the edge $(c_j^1,c_h^2)$ (see the dashed edges in the figure). Finally, for each clause $c_j$, $V$ also contains the nodes $d_j^i$ and $e_j^i$, for $i\in[K]$ (the value of $K$ will be specified later in the proof), and $E$ contains the edges $(d_j^i,c_j^1)$ and $(c_j^2,e_j^i)$, for $i\in[K]$.

\begin{figure}[t]
    \centerfloat{\includegraphics[scale=1]{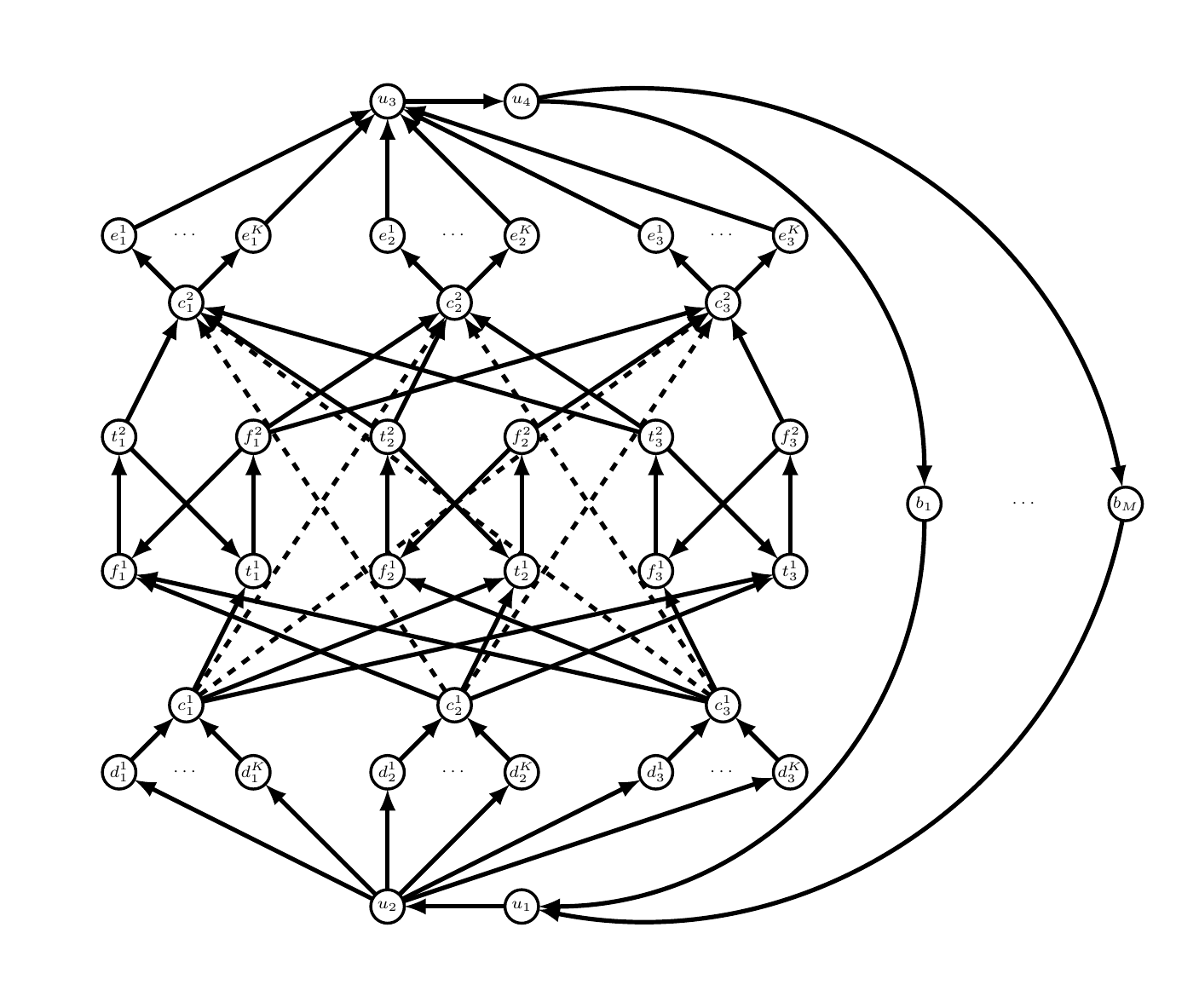}}
    \caption{An example of the reduction of \tsat{} to \mret{}. The \tsat{} formula is $(x_1\vee x_2\vee x_3)\wedge(\neg x_1\vee x_2\vee x_3)\wedge(\neg x_1\vee\neg x_2\vee\neg x_3)$.}
    \label{fig:mret-reduction}
\end{figure}

\item[\textbf{Block gadget}] $V$ contains the nodes $u_1$, $u_2$, $u_3$, and $u_4$, and the nodes $b_i$, for $i\in[M]$ (the value of $M$ will be specified later in the proof). $E$ contains the edges $\{(b_i,u_1) : i\in[M]\}$, $(u_1,u_2)$, $\{(u_2, d_j^l) : j\in[m], l\in[K]\}$, $\{(e_j^l,u_3) : j\in[m], l\in[K]\}$, $(u_3,u_4)$, and $\{(u_4,b_i) : i\in[M]\}$.
\end{description}

Note that $D$ is strongly connected. Indeed, let us consider the cycles
\begin{equation*}
    C_{i,j,l,p} = \langle u_1, u_2, d_j^l, c_j^1, t_i^1, f_i^2, f_i^1, t_i^2, c_j^2, e_j^l, u_3, u_4, b_p, u_1 \rangle,
\end{equation*}
where $j\in[m]$, $i$ is such that $x_i$ is a literal of the clause $c_j$, $l\in[K]$, and $p\in[M]$, and the cycles 
\begin{equation*}
    C_{i,j,l,p} = \langle u_1, u_2, d_j^l, c_j^1, f_i^1, t_i^2,t_i^1, f_i^2, c_j^2, e_j^l, u_3, u_4, b_p, u_1 \rangle,
\end{equation*}
where $j\in[m]$, $i$ is such that $\neg x_i$ is a literal of the clause $c_j$, $l\in[K]$, and $p\in[M]$. The union of these cycles contains each node in $V$, and each of these cycles contains node $u_1$. This proves the strong connectivity of $D$.

In the following, $B$ denotes the set $\{b_i : i\in[M]\}$ and $H$ denotes the set of nodes which do not belong to the block gadget, that is, $H=V\setminus(\{u_1,u_2,u_3,u_4\}\cup B)$ (note that $|V|=M+|H|+4$ and that $|H|=2(K+1)m+4n$).

\medskip
\noindent\textbf{Activation of pairs of nodes in a variable gadget}
Consider a variable $x_i$ and the associated variable gadget and a schedule $S$ of the 4 edges associated to the variable gadget. We say that $S$ \textit{activates} the pair $(t_i^1,t_i^2)$ (respectively, $(f_i^1,f_i^2)$) if $t_i^2$ (respectively, $f_i^2$) is $S$-reachable from $t_i^1$ (respectively, $f_i^1$) within the gadget, that is $(t_i^1,f_i^2)$, $(f_i^2,f_i^1)$, and $(f_i^1,t_i^2)$ (respectively, $(f_i^1,t_i^2)$, $(t_i^2,t_i^1)$, $(t_i^1,f_i^2)$) are scheduled in that order. Note that no schedule can activate both $(t_i^1,t_i^2)$ and $(f_i^1,f_i^2)$ as the edge $(t_i^1, f_i^2)$ is scheduled either before or after the edge $(f_i^1, t_i^2)$.

\medskip
\noindent\textbf{Constructing a schedule from a satisfying assignment.} Suppose that there is an assignment $\alpha$ that satisfies $\Phi$, and let us consider the following schedule $S$. First we schedule the edges in $\{(b_i, u_1) : i\in[M]\}$ (in any arbitrary order), then the edge $(u_1, u_2)$, and then the edges in $\{(u_2, d_j^i) : j\in[m], i\in[K]\}$ (in any arbitrary order). Then we schedule the edges $\{(d_j^i,c_j^1) : j\in[m], i\in[K]\}$ (in any arbitrary order), and, then, the edges going out from the nodes $c_j^1$, for $j\in[m]$ (in any arbitrary order). Then, for each $i\in[n]$, if $\alpha(x_i)=\true$, we schedule the edges $(t_i^1,f_i^2)$, $(f_i^2,f_i^1)$, $(f_i^1,t_i^2)$, and $(t_i^2,t_i^1)$ in this order (thus activating $(t_i^1,t_i^2)$). Otherwise (that is, $\alpha(x_i)=\false$), we schedule the edges $(f_i^1,t_i^2)$, $(t_i^2,t_i^1)$, $(t_i^1,f_i^2)$, and $(f_i^2,f_i^1)$ in this order (thus activating $(f_i^1,f_i^2)$). Then we schedule all the edges entering the nodes $c_j^2$, for $j\in[m]$ (in any arbitrary order), and then all the edges going out from the nodes $c_j^2$, for $j\in[m]$ (in any arbitrary order). Finally, we schedule the edges in $\{(e_j^l,u_3) : j\in[m],l\in[K]\}$ (in any arbitrary order), then the edge $(u_3,u_4)$, and all the edges in $\{(u_4,b_i) : i\in[M]\}$ (in any arbitrary order). Let $G$ be the temporal graph induced by $D$ and the schedule $S$.

First observe that, for any clause $c_j$ with $j\in[m]$, there exists a literal that satisfies $c_j$ according to the assignment $\alpha$. Let $x_i$ (respectively, $\neg x_i$) be a literal satisfying $c_j$. Since $(t_i^1,t_i^2)$ (respectively, $(f_i^1,f_i^2)$) is activated, there exists a temporal path from $c_j^1$ to $c_j^2$ that goes through variable gadget corresponding to $x_i$. This means that $c_j^2\in\reach{G}{c_j^1}$ and that, for $l,l'\in[K]$, $e_j^{l'}\in\reach{G}{d_j^l}$. We now prove a lower bound $\mathbb{L}$ on the $S$-reachability by showing a lower bound on the number of nodes temporally reachable from each possible source.

\begin{itemize}
    \item For any $v\in V$ and for $i\in[M]$, $v\in\reach{G}{b_i}$. This adds $M(M+|H|+4)$ to $\mathbb{L}$.
    
    \item For $i\in[4]$ and for $j\in[M]$, $b_j\in\reach{G}{u_i}$. Moreover, $\{u_1,u_2,u_3,u_4\}\cup H\subseteq\reach{G}{u_1}$, $\{u_2,u_3,u_4\}\cup H\subseteq\reach{G}{u_2}$, $u_3,u_4\in\reach{G}{u_3}$, and $u_4\in\reach{G}{u_4}$. This adds $4M+2|H|+10$ to $\mathbb{L}$.
    
    \item For $j,h\in[m]$, $i,l\in[K]$, and $p\in[M]$, $c_h^2,e_h^l\in\reach{G}{d_j^i}$ (because of the above observation) and $b_p\in\reach{G}{d_j^i}$. This adds $Km(M+Km+m)$ to $\mathbb{L}$.
    
    \item For $j,h\in[m]$, $l\in[K]$, and $i\in[M]$, $c_h^2,e_h^l,b_i\in\reach{G}{c_j^1}$. This adds $m(M+Km+m)$ to $\mathbb{L}$.
    
    \item For $i\in[n]$, there exists $j\in[m]$ such that $c_j$ is satisfied by $\alpha(x_i)$. Hence, for $p\in[2]$, $l\in[K]$, and $h\in[M]$,  $e_j^l,b_h\in\reach{G}{t_i^p}$ and $e_j^l,b_h\in\reach{G}{f_i^p}$. This adds $4n(M+K)$ to $\mathbb{L}$.
    
    \item For $j\in[m]$, $l\in[K]$, and $h\in[M]$, $e_j^l,b_h\in\reach{G}{c_j^2}$. This adds $m(M+K)$ to $\mathbb{L}$.
    
    \item For $j\in[m]$, $l\in[K]$, and $h\in[M]$, $b_h\in\reach{G}{e_j^l}$. This adds $MKm$ to $\mathbb{L}$.
\end{itemize}
Thus, the $S$-reachability is at least
\begin{eqnarray*}
\mathbb{L} & = & M(M+|H|+4)+(4M+2|H|+10)+Km(M+Km+m)\\
  &   & +m(M+Km+m)+4n(M+K)+m(M+K)+MKm.
\end{eqnarray*}

\medskip
\noindent\textbf{Bounding reachability when $\Phi$ is not satisfiable.} Let us set $M$ equal to any value greater than $(|H|+5)^2$. We now prove that, if there exists no truth-assignment satisfying the formula $\Phi$, then no schedule $S$ can have $S$-reachability greater than or equal to $\mathbb{L}$. First notice that if $S$ assigns to the edge $(u_3,u_4)$ a starting time smaller than the starting time assigned to $(u_1,u_2)$, then the $S$-reachability is less than $\mathbb{L}$. This is because, in this case, for $i,j\in[M]$ with $i\neq j$, $b_j$ is not $S$-reachable from $b_i$. Hence, the $S$-reachability is bounded by $\mathbb{U}_1=M(|H|+4+1)+(|H|+4)(M+|H|+4)$: this would happen if, for each node $v\not\in B$, $\reach{G}{v}=V$. Since $\mathbb{L}>M^2$, $\mathbb{U}_1=M(|H|+4+1)+(|H|+4)(M+|H|+4)=2M(|H|+4)+(|H|+4)^2+M < M(|H|+5)^2$, and $M>(|H|+5)^2$, it holds that $\mathbb{L}>\mathbb{U}_1$. We can then focus on schedules that assign to the edge $(u_1,u_2)$ a starting time smaller than the starting time assigned to the edge $(u_3,u_4)$. Let $S$ be such a schedule and let $G$ be the temporal graph induced by $D$ and $S$. We now prove an upper bound $\mathbb{U}_2$ on the $S$-reachability by giving an upper bound on the nodes reachable from each possible source. Observe that, for any two nodes $u$ and $v$, $v$ might belong to $\reach{G}{u}$ only if in $D$ there exists a path from $u$ to $v$ that does not include the edge $(u_3,u_4)$ before the edge $(u_1,u_2)$.

\begin{itemize}
    \item For $i\in[M]$, $|\reach{G}{b_i}|\leq|V|$. This adds $M(M+|H|+4)$ to $\mathbb{U}_2$.
    
    \item For $i\in[2]$, $|\reach{G}{u_i}|\leq|V|$, while $|\reach{G}{u_3}|\leq M+3$ and $|\reach{G}{u_4}|\leq|V|=M+|H|+4$. 
    
    \item For $j\in[m]$ and $i\in[K]$, in the best case $\reach{G}{d_j^i}$ contains $d_j^i$, $c_j^1$, the 12 nodes corresponding to the three variables appearing in $c_j$, and the nodes in $\{c_h^2 : h\in[m]\} \cup \{e_h^l : h\in[m], l\in[K]\} \cup \{u_1,u_3,u_4\} \cup B$, yielding $|\reach{G}{d_j^i}|\le M+Km+m+17$. However, we can show that there exists an index $j^*$ such that, for $l,l'\in[K]$, $e_{j^*}^{l'}\not\in\reach{G}{d_{j^*}^l}$, implying that the $d$-nodes add at most $Km(M+Km+m+17)-K^2$ to $\mathbb{U}_2$. For defining $j^*$, we consider the following truth-assignment $\alpha$: for any variable $x_i$ with $i\in[n]$, $\alpha(x_i)=\true$ if $(t_i^1, f_i^2)$ is scheduled before $(f_i^1, t_i^2)$, otherwise $\alpha(x_i)=\false$. Note that if $\alpha(x_i)=\true$ (respectively, $\alpha(x_i)=\false$) we know that $S$ does not activate $(f_i^1, f_i^2)$ (respectively, $(t_i^1, t_i^2)$). Since the formula $\Phi$ is not satisfiable there exists $j^*\in[m]$ such that $c_{j^*}$ is not satisfied by $\alpha$. Let $x_i$ (respectively, $\neg x_i$) be a literal in $c_{j^*}$. Since $c_{j^*}$ is not satisfied by $\alpha$, we that $\alpha(x_i)=\false$ (respectively, $\alpha(x_i)=\true$) and that $(t_i^1,t_i^2)$ (respectively, $(f_i^1,f_i^2)$) is not activated. It is thus impossible to reach $c_{j^*}^2$ from $c_{j^*}^1$ through the variable gadget of $x_i$. On the other hand, in all the other walks in $D$ that connect $c_{j^*}^1$ to $c_{j^*}^2$ the edge $(u_3,u_4)$ appears before the edge $(u_1,u_2)$. Hence, $c_{j^*}^2\not\in\reach{G}{c_{j^*}^1}$ and $e_{j^*}^{l'}\not\in\reach{G}{d_{j^*}^l}$ for $l,l'\in[K]$.
    
    \item For $j\in[m]$, in the best case $\reach{G}{c_j^1}$ contains $c_j^1$, the 12 nodes corresponding to the three variables appearing in clause $c_j$, and the nodes in $\{c_h^2 : h\in[m]\} \cup \{e_h^l : h\in[m], l\in[K]\} \cup \{u_1,u_3,u_4\} \cup B$. This adds $m(M+Km+m+16)$ to $\mathbb{U}_2$.
    
    \item For $i\in[n]$ and for $j\in[2]$, in the best case $\reach{G}{t_i^j}$ and $\reach{G}{f_i^j}$ contain the corresponding four variable nodes and the nodes in $\{c_h^2 : h\in[m]\} \cup \{e_h^l : h\in[m], l\in[K]\} \cup \{u_1,u_3,u_4\} \cup B$. This adds $4n(M+Km+m+7)$ to $\mathbb{U}_2$.
    
    \item For $j\in[m]$, in the best case $\reach{G}{c_j^2}$ contains $c_j^2$ and the nodes in $\{e_j^l : l\in[K]\} \cup \{u_1,u_3,u_4\} \cup B$. This adds $m(M+K+4)$ to $\mathbb{U}_2$.
    
    \item For $j\in[m]$ and $i\in[K]$, in the best case $\reach{G}{e_j^i}$ contains $e_j^i$ and the nodes in $\{u_1,u_3,u_4\} \cup B$. This adds with $Km(M+4)$ to $\mathbb{U}_2$. 
\end{itemize}
In summary,
\begin{eqnarray*}
\mathbb{U}_2 & = & M(M+|H|+4)+(4M+3|H|+15)+(Km(M+Km+m+17)-K^2)\\
             &   & +m(M+Km+m+16)+4n(M+Km+m+7)+m(M+K+4)\\
             &   & +Km(M+4).
\end{eqnarray*}
We have that $\mathbb{L}-\mathbb{U}_2 = -|H|-5 K^2-21Km-4n(K(m-1)+m+7)-20m = K^2-23Km-4n(K(m-1)+m+8)-22m-5 > K^2-Knm(23+4(1+1+8)+22+5)=K^2-90Knm$ using $K,n,m\ge 1$. Let us set $K$ equal to any value greater than or equal to $91nm$. We then have $K^2 > 90Knm$ and, thus, $\mathbb{L}>\mathbb{U}_2$. That is, the $S$-reachability has to be smaller than $\mathbb{L}$.

\medskip
\noindent\textbf{Conclusion.} We have thus proved that the formula $\Phi$ is satisfiable if and only if there exists a schedule $S$ such that the $S$-reachability of $D$ is at least $\mathbb{L}$. This completes the proof of the theorem.\qed
\end{proof}

\section{Conclusion and open problems}

In this paper, we have considered \mret\ problem, that is, the problem of assigning appearing times to the edges of a digraph in order to maximize the total reachability of the resulting temporal graph. We have proved that this problem is \nptime-hard, even when the digraph is strongly connected. We conjecture that the \mret\ problem can be approximated within a constant approximation ratio. In particular, we conjecture that any strongly connected digraph admits an edge temporalisation with temporal reachability at least equal to $c\cdot n^2$ for some constant $c>0$. One way to prove such a statement would be to prove the following interesting graph theory conjecture.

\medskip

\noindent\textbf{Almost Spanning Two Rooted-Arborescences conjecture ({\sc astra}).} \textit{Any strongly connected digraph admits an out-arborescence and an in-arborescence that are edge-disjoint, have the same root, and each spans $\Omega(n)$ nodes.}

\medskip

\begin{figure}[b]
    \centerfloat{\includegraphics[scale=1]{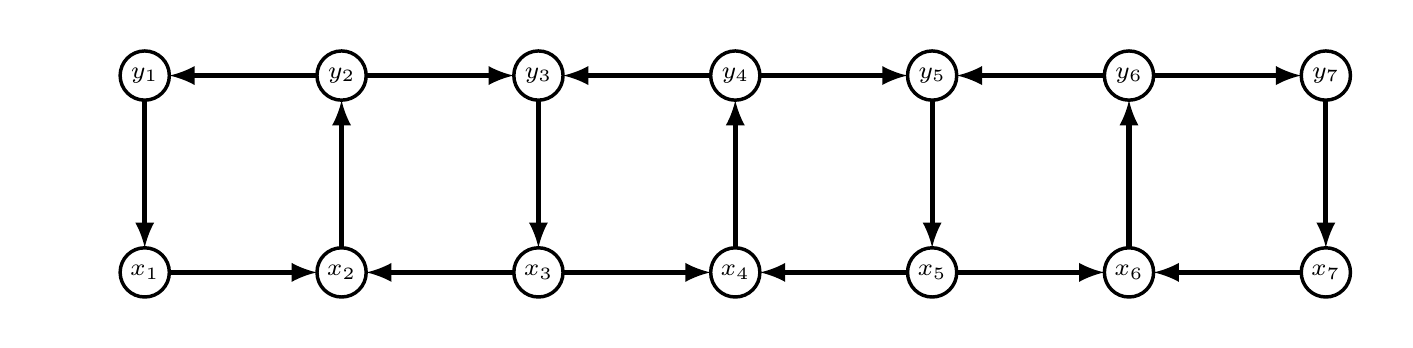}}
    \caption{An example of a digraph for which only some nodes can be roots of two arborescences each spanning $\Omega(n)$ nodes.}
    \label{fig:notallroots}
\end{figure}

Note that it is not difficult to prove that the root of the two arborescences mentioned in the \textsc{astra} conjecture cannot be any node in the graph. For example, let us consider the graph shown in Figure~\ref{fig:notallroots}. In this case, the node $x_{1}$ cannot be the common root of the two arborescences, since the only in-arborescence and the only out-arborescence with root $x_{1}$ share the edge $(x_{2}, y_{2})$, so that one of the two arborescences cannot include more than one node (of course, this example can be generalized to any even number of nodes).

\begin{figure}[t]
    \centerfloat{\includegraphics[scale=0.3]{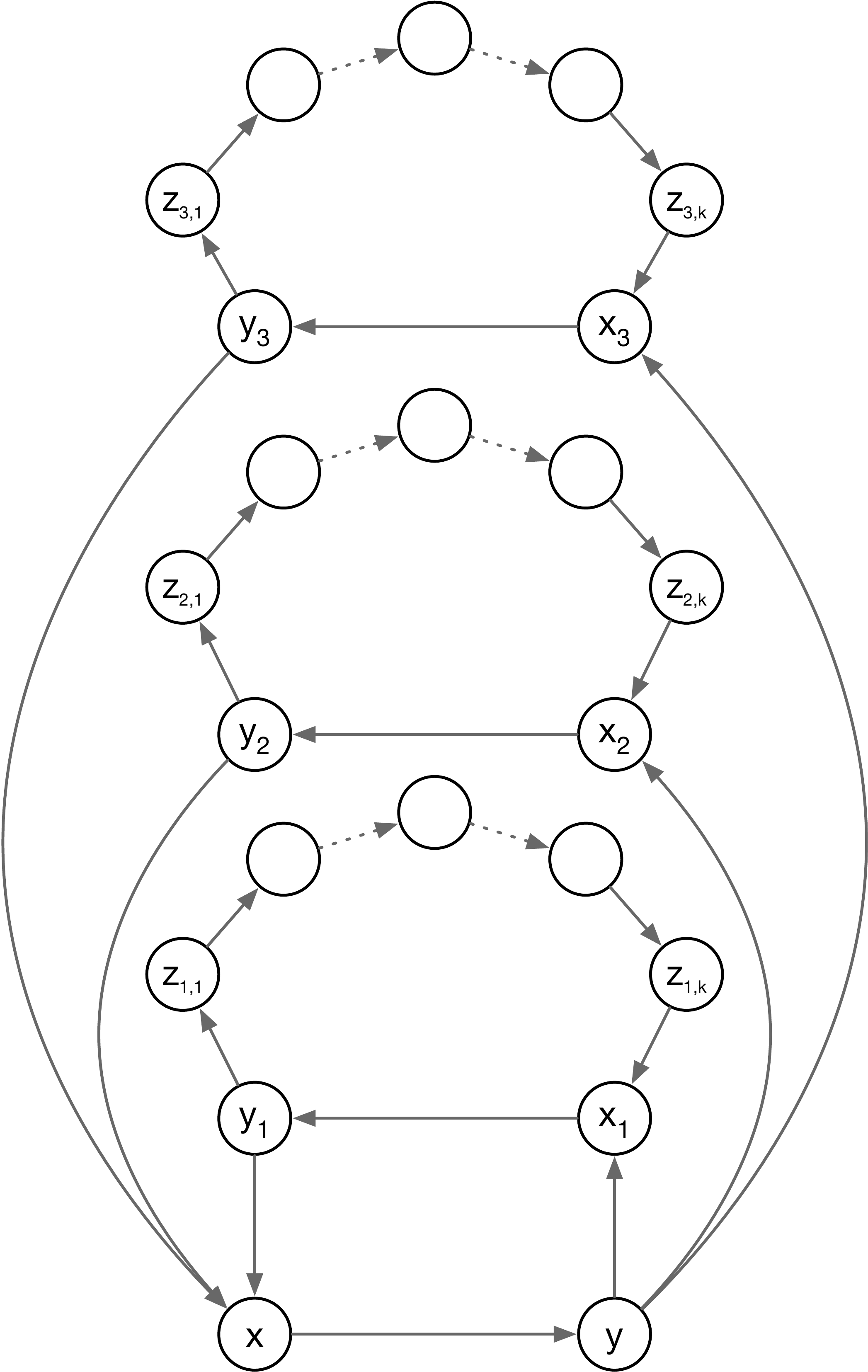}}
    \caption{An example of a digraph for which there are no two edge-disjoint arborescences with a common root and each spanning more than $n/3 + c$ nodes, for some positive constant $c$.}
    \label{fig:nothalfnodes}
\end{figure}

Note also that the \textsc{astra} conjecture is false if we require that the two arborescences span at least $\frac{n}{3-\epsilon}$ nodes, for any positive constant $\epsilon$.
For example, consider the digraph $G = (V,E)$ shown in Figure~\ref{fig:nothalfnodes}, where, for some integer parameter $k > 0$, the set of nodes is $V = \{x,y,x_1,y_1,x_2,y_2,x_3,y_3\} \cup \{z_{i,j} ~|~ 1 \le i \le 3, 1 \le j \le k\}$, and the set of edges is $E = \{(x,y)\} \cup \{(y,x_i),(x_i,y_i), (y_i,x), (y_i,z_{i,1}), (z_{i,k},x_i) ~|~ 1 \le i \le 3\} \cup \{(z_{i,j},z_{i,j+1}) ~|~ 1 \le i \le 3, 1 \le j < k\}$. Observe that the total number of nodes is $n = 3 k + 8$. Let us first give an upper bound on the minimum between the amount of nodes in the in-arborescence and in the out-arborescence in the case where the root is not $x$ nor $y$. For any such node, either the in-arborescence or the out-arborescence can contain at most $k+3$ nodes, since the edge $(x,y)$ can be in one arborescence only. Consider now the case in which either $x$ or $y$ is the root. 
Let us suppose that the root is $x$ (the other case can be analysed in a similar way). Since, for each $i=1,2,3$, the edge $(x_{i},y_{i})$ can be in one arborescence only, then either the in-arborescence or the out-arborescence rooted at $x$ can contain at most $n - 2k = k + 8$ nodes. We thus obtained that, in all cases, either the in-arborescence or the out-arborescence is upper bounded by $k+8 = n/3 + O(1)$.

\bibliographystyle{plain}
\bibliography{biblio}

\end{document}